\documentclass[11pt]{amsart}

\usepackage{amsmath,amssymb,amsthm}
\usepackage{hyperref}
\usepackage{graphicx}
\usepackage[margin=1in]{geometry}
\usepackage{booktabs}
\usepackage{array}
\usepackage{enumitem}

\makeatletter
\@namedef{subjclassname@2020}{%
  \textup{Subject Classification (MSC 2020)}}
\makeatother
\makeatletter
\renewcommand{\section}{\@startsection{section}{1}%
  \z@{.7\linespacing\@plus\linespacing}{.5\linespacing}%
  {\normalfont\bfseries}}
\renewcommand{\subsection}{\@startsection{subsection}{2}%
  \z@{.5\linespacing\@plus.7\linespacing}{-.5em}%
  {\normalfont\bfseries}}
\makeatother

\theoremstyle{plain}
\newtheorem{theorem}{Theorem}[section]
\newtheorem{lemma}[theorem]{Lemma}
\newtheorem{proposition}[theorem]{Proposition}

\theoremstyle{definition}
\newtheorem{definition}[theorem]{Definition}
\newtheorem{example}[theorem]{Example}

\theoremstyle{remark}
\newtheorem{remark}[theorem]{Remark}

\newcommand{\gc}[1]{\ulcorner #1 \urcorner}

\begin{document}

\title{A Linear-Size Block-Partition Fibonacci Encoding\\for G\"odel Numbering}

\author{Zolt\'an S\'ostai}
\thanks{ORCID: \href{https://orcid.org/0009-0007-8628-4276}{0009-0007-8628-4276}}

\date{March 2026}

\begin{abstract}
We construct an encoding of finite strings over a fixed finite alphabet as natural numbers, based on a block partition of the Fibonacci sequence. Each position in the string selects one Fibonacci number from a dedicated block, with unused indices between blocks guaranteeing non-adjacency. The encoded number is the sum of the selected Fibonacci numbers, and Zeckendorf's theorem \cite{Zeckendorf1972,Lekkerkerker1952} guarantees that this sum uniquely determines the selection. The encoding is injective, the string length is recoverable from the code, and the worst-case digit count of the encoded number grows as $\Theta(m)$ for strings of length~$m$---matching the information-theoretic lower bound up to a constant factor. We also prove that the natural right-nested use of Rosko's~\cite{Rosko2025} binary carryless pairing for sequence encoding has worst-case $\Theta(2^m)$ digit growth---an exponential blowup that the block-partition construction avoids entirely.
\end{abstract}



\maketitle

\section{Introduction}

G\"odel's incompleteness theorems~\cite{Goedel1931} require encoding syntactic objects---symbols, formulas, proofs---as natural numbers. The original encoding maps a sequence $\langle a_1, \ldots, a_m \rangle$ to the product
\[
2^{a_1+1} \cdot 3^{a_2+1} \cdot 5^{a_3+1} \cdots p_m^{a_m+1},
\]
where $p_i$ is the $i$-th prime and the~$+1$ ensures positive exponents. The fundamental theorem of arithmetic guarantees injectivity, and decoding reduces to factorization. This encoding produces very large numbers: even short formulas yield codes of dozens or hundreds of digits.

Rosko~\cite{Rosko2025} proposed an alternative based on Zeckendorf representations---unique decompositions of natural numbers as sums of non-consecutive Fibonacci numbers~\cite{Zeckendorf1972,Lekkerkerker1952}. His ``carryless pairing'' function is a pairing $\mathbb{N}^2 \to \mathbb{N}$ that places the Zeckendorf supports of two inputs into disjoint even/odd index bands. Encoding sequences of length~$m$ via right-nesting of this binary pairing requires $m-1$ applications. We prove in~\S\ref{sec:blowup} that this right-nested use has worst-case $\Theta(2^m)$ digit growth.

The present paper introduces a simpler construction: a \textbf{block-partition encoding} that partitions the Fibonacci sequence into contiguous blocks, one per string position, with single-index gaps between blocks to ensure non-adjacency. Each symbol selects one Fibonacci number from its position's block, and the encoded number is the sum. No nesting is involved. The digit count of the encoded number grows linearly in the string length.

The paper isolates the \textbf{numeric coding layer} of a G\"odel numbering: the injective mapping from strings to natural numbers. The further machinery required for a full syntactic arithmetization---substitution, concatenation, diagonalization---is not addressed here.

\begin{remark}\label{rem:fibonacci-coding}
The term ``Fibonacci coding'' already denotes a variable-length binary code used in data compression~\cite{Fraenkel1985}, where each positive integer is represented by a binary string derived from its Zeckendorf representation. Our construction is distinct: we encode \emph{strings over a formal alphabet} as single natural numbers. The shared foundation is Zeckendorf's theorem, but the purpose and mechanism differ.
\end{remark}

\section{Preliminaries}

\subsection{Fibonacci sequence}

Let $\{F_n\}_{n \ge 0}$ denote the Fibonacci sequence defined by $F_0=0$, $F_1=1$, and $F_{n+2}=F_{n+1}+F_n$ for all $n \ge 0$. The sequence begins
\[
0,\; 1,\; 1,\; 2,\; 3,\; 5,\; 8,\; 13,\; 21,\; 34,\; 55,\; 89,\; 144,\; 233,\; 377,\; 610,\; 987,\; \ldots
\]
It is well known that $F_n = (\varphi^n - \psi^n)/\sqrt{5}$, where $\varphi = (1+\sqrt{5})/2 \approx 1.618$ is the golden ratio and $\psi = (1-\sqrt{5})/2 \approx -0.618$ (see, e.g., \cite{GKP1994}, \S6.6).

\subsection{Zeckendorf's theorem}

\begin{theorem}[Zeckendorf~\cite{Zeckendorf1972}, Lekkerkerker~\cite{Lekkerkerker1952}]\label{thm:zeck}
Every positive integer~$n$ has a unique representation as a sum of non-consecutive Fibonacci numbers:
\[
n = F_{e_1} + F_{e_2} + \cdots + F_{e_s}
\]
where $e_1 > e_2 > \cdots > e_s \ge 2$ and $e_i - e_{i+1} \ge 2$ for all~$i$.
\end{theorem}

The set $Z(n) = \{e_1, e_2, \ldots, e_s\}$ is called the \textbf{Zeckendorf support} of~$n$. We define $Z(0)=\varnothing$. The support of a positive integer is computed by the \textbf{greedy algorithm}: at each step, select the largest Fibonacci number not exceeding the remainder, subtract it, and repeat.

\section{The block-partition encoding}

\subsection{The alphabet}

Let $\Sigma = \{\sigma_0, \sigma_1, \ldots, \sigma_{k-1}\}$ be a finite alphabet of $k \ge 2$ symbols. Each symbol~$\sigma_j$ is assigned the \textbf{symbol code} $j \in \{0, 1, \ldots, k-1\}$.

For illustration, we use a 10-symbol alphabet drawn from Peano arithmetic:

\medskip
\begin{center}
\begin{tabular}{c|cccccccccc}
\toprule
Code & 0 & 1 & 2 & 3 & 4 & 5 & 6 & 7 & 8 & 9 \\
\midrule
Symbol & $\mathtt{0}$ & $S$ & $+$ & $\cdot$ & $=$ & $\lnot$ & $\to$ & $\forall$ & $($ & $)$ \\
\bottomrule
\end{tabular}
\end{center}
\medskip

A complete first-order syntax would require a finite tokenization scheme for variables (e.g., de~Bruijn indices), proof delimiters, and possibly additional connectives. The construction handles any finite alphabet identically by adjusting~$k$.

\subsection{Block partition}

\begin{definition}[Block partition]\label{def:block}
For a $k$-symbol alphabet, define the \textbf{block size} $\beta = k+1$. For each position $i \in \{0, 1, 2, \ldots\}$ in a string, define the \textbf{position block} $B_i$ as the set of Fibonacci indices:
\[
B_i = \bigl\{\, 2 + i\beta,\; 2 + i\beta + 1,\; \ldots,\; 2 + i\beta + (k-1) \,\bigr\}.
\]
Each block $B_i$ contains exactly~$k$ consecutive Fibonacci indices---one for each symbol. The index $2 + i\beta + k$ (the first index of the gap) is \textbf{unused}, ensuring that the last index of~$B_i$ and the first index of~$B_{i+1}$ are separated by at least~$2$.
\end{definition}

For the 10-symbol alphabet ($k=10$, $\beta=11$), the encoding is displayed in Table~\ref{tab:encoding}. Each row corresponds to one symbol and shows the Fibonacci sequence from $F_2$ onward. \textbf{Bold} numbers are the Fibonacci numbers assigned to that symbol at each successive position; \textit{italicized} numbers are the unused gap entries. The three blocks shown correspond to positions~1, 2, and~3.

\begin{table}[ht]
\centering
\resizebox{\textwidth}{!}{%
\begin{tabular}{r@{\;:\;\,}*{33}{r}}
\toprule
& \multicolumn{10}{c}{\footnotesize Block $B_0$ (Pos.\,1)} & \multicolumn{1}{c}{\footnotesize\textit{gap}} & \multicolumn{10}{c}{\footnotesize Block $B_1$ (Pos.\,2)} & \multicolumn{1}{c}{\footnotesize\textit{gap}} & \multicolumn{10}{c}{\footnotesize Block $B_2$ (Pos.\,3)} & \multicolumn{1}{c}{\footnotesize\textit{gap}} \\
\cmidrule(lr){2-11}\cmidrule(lr){12-12}\cmidrule(lr){13-22}\cmidrule(lr){23-23}\cmidrule(lr){24-33}\cmidrule(lr){34-34}
$\mathtt{0}$ & \textbf{1} & 2 & 3 & 5 & 8 & 13 & 21 & 34 & 55 & 89 & \textit{144} & \textbf{233} & 377 & 610 & 987 & 1597 & 2584 & 4181 & 6765 & 10946 & 17711 & \textit{28657} & \textbf{46368} & 75025 & 121393 & 196418 & 317811 & 514229 & 832040 & 1346269 & 2178309 & 3524578 & \textit{5702887} \\
$S$ & 1 & \textbf{2} & 3 & 5 & 8 & 13 & 21 & 34 & 55 & 89 & \textit{144} & 233 & \textbf{377} & 610 & 987 & 1597 & 2584 & 4181 & 6765 & 10946 & 17711 & \textit{28657} & 46368 & \textbf{75025} & 121393 & 196418 & 317811 & 514229 & 832040 & 1346269 & 2178309 & 3524578 & \textit{5702887} \\
$+$ & 1 & 2 & \textbf{3} & 5 & 8 & 13 & 21 & 34 & 55 & 89 & \textit{144} & 233 & 377 & \textbf{610} & 987 & 1597 & 2584 & 4181 & 6765 & 10946 & 17711 & \textit{28657} & 46368 & 75025 & \textbf{121393} & 196418 & 317811 & 514229 & 832040 & 1346269 & 2178309 & 3524578 & \textit{5702887} \\
$\cdot$ & 1 & 2 & 3 & \textbf{5} & 8 & 13 & 21 & 34 & 55 & 89 & \textit{144} & 233 & 377 & 610 & \textbf{987} & 1597 & 2584 & 4181 & 6765 & 10946 & 17711 & \textit{28657} & 46368 & 75025 & 121393 & \textbf{196418} & 317811 & 514229 & 832040 & 1346269 & 2178309 & 3524578 & \textit{5702887} \\
$=$ & 1 & 2 & 3 & 5 & \textbf{8} & 13 & 21 & 34 & 55 & 89 & \textit{144} & 233 & 377 & 610 & 987 & \textbf{1597} & 2584 & 4181 & 6765 & 10946 & 17711 & \textit{28657} & 46368 & 75025 & 121393 & 196418 & \textbf{317811} & 514229 & 832040 & 1346269 & 2178309 & 3524578 & \textit{5702887} \\
$\lnot$ & 1 & 2 & 3 & 5 & 8 & \textbf{13} & 21 & 34 & 55 & 89 & \textit{144} & 233 & 377 & 610 & 987 & 1597 & \textbf{2584} & 4181 & 6765 & 10946 & 17711 & \textit{28657} & 46368 & 75025 & 121393 & 196418 & 317811 & \textbf{514229} & 832040 & 1346269 & 2178309 & 3524578 & \textit{5702887} \\
$\to$ & 1 & 2 & 3 & 5 & 8 & 13 & \textbf{21} & 34 & 55 & 89 & \textit{144} & 233 & 377 & 610 & 987 & 1597 & 2584 & \textbf{4181} & 6765 & 10946 & 17711 & \textit{28657} & 46368 & 75025 & 121393 & 196418 & 317811 & 514229 & \textbf{832040} & 1346269 & 2178309 & 3524578 & \textit{5702887} \\
$\forall$ & 1 & 2 & 3 & 5 & 8 & 13 & 21 & \textbf{34} & 55 & 89 & \textit{144} & 233 & 377 & 610 & 987 & 1597 & 2584 & 4181 & \textbf{6765} & 10946 & 17711 & \textit{28657} & 46368 & 75025 & 121393 & 196418 & 317811 & 514229 & 832040 & \textbf{1346269} & 2178309 & 3524578 & \textit{5702887} \\
$($ & 1 & 2 & 3 & 5 & 8 & 13 & 21 & 34 & \textbf{55} & 89 & \textit{144} & 233 & 377 & 610 & 987 & 1597 & 2584 & 4181 & 6765 & \textbf{10946} & 17711 & \textit{28657} & 46368 & 75025 & 121393 & 196418 & 317811 & 514229 & 832040 & 1346269 & \textbf{2178309} & 3524578 & \textit{5702887} \\
$)$ & 1 & 2 & 3 & 5 & 8 & 13 & 21 & 34 & 55 & \textbf{89} & \textit{144} & 233 & 377 & 610 & 987 & 1597 & 2584 & 4181 & 6765 & 10946 & \textbf{17711} & \textit{28657} & 46368 & 75025 & 121393 & 196418 & 317811 & 514229 & 832040 & 1346269 & 2178309 & \textbf{3524578} & \textit{5702887} \\
\bottomrule
\end{tabular}%
}

\vspace{1em}
\caption{Block-partition mapping for the 10-symbol alphabet. Each row shows the Fibonacci sequence for one symbol, with \textbf{bold} entries marking the assigned Fibonacci number at each position (blocks of~10, separated by italicized gap entries). To encode a string: for each position, find the row of the symbol at that position, read off its bold entry for that block, and sum all selected entries. Zeckendorf's theorem guarantees that the sum uniquely determines the selection.}
\label{tab:encoding}
\end{table}

\subsection{Encoding}

\begin{definition}[Block-partition encoding]\label{def:encoding}
Let $s = s_1 s_2 \cdots s_m$ be a nonempty string of length~$m$ over~$\Sigma$, with symbol codes $c_1, c_2, \ldots, c_m \in \{0,1,\ldots,k-1\}$. The \textbf{block-partition Fibonacci code} of~$s$ is
\[
\gc{s} \;=\; \sum_{i=0}^{m-1} F_{2 + i\beta + c_{i+1}}.
\]
That is, for position~$i$ (0-indexed), select the Fibonacci number at index $2 + i\beta + c_{i+1}$ from block~$B_i$. For the empty string, define $\gc{\varepsilon} = 0$.
\end{definition}

\begin{example}\label{ex:0eq0}
The string $\mathtt{0=0}$ has symbol codes $(0,4,0)$. Under the encoding:

\smallskip
\noindent
\begin{tabular}{@{}l@{}}
Position~0, symbol~$\mathtt{0}$ (code~0): $F_2 = 1$ \\
Position~1, symbol~$=$ (code~4): $F_{17} = 1597$ \\
Position~2, symbol~$\mathtt{0}$ (code~0): $F_{24} = 46368$
\end{tabular}
\smallskip

\noindent Hence $\gc{\mathtt{0=0}} = 1 + 1597 + 46368 = \mathbf{47966}$.
\end{example}

\begin{example}\label{ex:s0eqs0}
The string $S(\mathtt{0})=S(\mathtt{0})$ has symbol codes $(1,8,0,9,4,1,8,0,9)$. Under the encoding:

\smallskip
\noindent
\begin{tabular}{@{}l@{}}
Position~0, symbol~$S$ (code~1): $F_3 = 2$ \\
Position~1, symbol~$($ (code~8): $F_{21} = 10946$ \\
Position~2, symbol~$\mathtt{0}$ (code~0): $F_{24} = 46368$ \\
Position~3, symbol~$)$ (code~9): $F_{44} = 701408733$ \\
Position~4, symbol~$=$ (code~4): $F_{50} = 12586269025$ \\
Position~5, symbol~$S$ (code~1): $F_{58} = 591286729879$ \\
Position~6, symbol~$($ (code~8): $F_{76} = 3416454622906707$ \\
Position~7, symbol~$\mathtt{0}$ (code~0): $F_{79} = 14472334024676221$ \\
Position~8, symbol~$)$ (code~9): $F_{99} = 218922995834555169026$
\end{tabular}
\smallskip

\noindent Hence $\gc{S(\mathtt{0})=S(\mathtt{0})} = 218940885227777216907$, a 21-digit number. For comparison, the prime-power G\"odel encoding of the same string produces a number of approximately 49~digits.
\end{example}

\section{Correctness}

\subsection{Non-adjacency}

\begin{lemma}[Non-adjacency]\label{lem:nonadj}
For any nonempty string $s = s_1 s_2 \cdots s_m$ over a $k$-symbol alphabet, the Fibonacci indices selected by the block-partition encoding are mutually non-consecutive.
\end{lemma}

\begin{proof}
Each position~$i$ selects exactly one index from block~$B_i$. Within a single block, only one index is selected, so no adjacency arises within a block.

Between consecutive blocks: the maximum index in~$B_i$ is $2 + i\beta + (k-1)$. The minimum index in~$B_{i+1}$ is $2 + (i+1)\beta = 2 + i\beta + k + 1$. The difference is
\[
(2 + i\beta + k + 1) - (2 + i\beta + k - 1) = 2,
\]
so the two indices are non-consecutive in the Fibonacci sequence. For non-adjacent blocks~$B_i$ and~$B_j$ with $j > i+1$, the gap is strictly larger.
\end{proof}

\subsection{Injectivity}

\begin{theorem}[Injectivity]\label{thm:inj}
The block-partition encoding is injective on~$\Sigma^*$: if $s \ne t$, then $\gc{s} \ne \gc{t}$.
\end{theorem}

\begin{proof}
If one of $s,t$ is the empty string and the other is not, then one code is~$0$ and the other is a positive integer (since every nonempty string selects at least one $F_e \ge F_2 = 1$).

If $s$ and~$t$ are both nonempty with $|s| = m < m' = |t|$, then the Zeckendorf support of~$\gc{t}$ contains an index from block~$B_m$ or higher, while every index in $Z(\gc{s})$ lies in blocks $B_0$ through~$B_{m-1}$. These are distinct subsets of non-consecutive indices (by Lemma~\ref{lem:nonadj}), so by Theorem~\ref{thm:zeck} the sums differ.

If $|s|=|t|$ but $s$ and~$t$ differ in at least one position~$j$, then position~$j$ selects a different Fibonacci index in~$s$ than in~$t$. The remaining positions select identical indices. So the two Zeckendorf supports differ, and by Theorem~\ref{thm:zeck} the sums differ.
\end{proof}

\subsection{Valid codes and length recovery}

\begin{definition}[Valid code]\label{def:valid}
A natural number~$n$ is a \textbf{valid code} if $n=0$ (encoding the empty string) or if there exist $m \ge 1$ and unique $c_1, c_2, \ldots, c_m \in \{0,1,\ldots,k-1\}$ such that
\[
Z(n) = \bigl\{\, 2 + i\beta + c_{i+1} : 0 \le i < m \,\bigr\}.
\]
That is, the Zeckendorf support of~$n$ consists of exactly~$m$ indices, one from each block~$B_0$ through~$B_{m-1}$, with no indices in gap positions. The image of the encoding~$\gc{\cdot}$ is exactly the set of valid codes.
\end{definition}

\begin{lemma}[Length recovery]\label{lem:length}
For any valid code~$n$, the length of the encoded string is recoverable: if $n=0$, then $m=0$; if $n > 0$, then $m = |Z(n)|$.
\end{lemma}

\begin{proof}
For $n > 0$, Definition~\ref{def:valid} requires $Z(n)$ to consist of exactly~$m$ indices, one per block. Therefore $|Z(n)|=m$. Equivalently, $m = \lfloor (e_{\max}-2)/\beta \rfloor + 1$ where $e_{\max} = \max Z(n)$, since the largest index lies in~$B_{m-1}$.
\end{proof}

\subsection{Decoding}

\begin{theorem}[Decoding]\label{thm:decode}
For any valid code~$n$, the original string~$s$ with $\gc{s}=n$ is uniquely recoverable.
\end{theorem}

\begin{proof}
If $n=0$, return the empty string. Otherwise, apply the greedy Zeckendorf algorithm to~$n$ to obtain the support~$Z(n)$. By Definition~\ref{def:valid} and Lemma~\ref{lem:length}, this support contains exactly~$m$ indices, one per block. For each recovered index~$e_j$, compute
\[
\text{position:}\quad i = \bigl\lfloor (e_j - 2)/\beta \bigr\rfloor, \qquad \text{symbol code:}\quad c = (e_j - 2) \bmod \beta.
\]
This recovers the symbol at each position. The division and modulus by the fixed constant $\beta = k+1$ are computable by repeated subtraction.
\end{proof}

\section{Complexity analysis}

\subsection{Code length}

\begin{theorem}[Linear code length]\label{thm:linear}
For $k \ge 2$ and $m \ge 1$, the worst-case digit count $L(m) = \max\{\operatorname{digits}(\gc{s}) : s \in \Sigma^m\}$ satisfies $L(m) = \Theta(m)$.
\end{theorem}

\begin{proof}
The largest Fibonacci index used in encoding a string of length~$m$ is at most
\[
e_{\max} = 2 + (m-1)\beta + (k-1) = m\beta.
\]
Every length-$m$ code satisfies
\[
F_{2+(m-1)\beta} \;\le\; \gc{s} \;<\; F_{m\beta+1}.
\]
The lower bound holds because the code includes at least $F_{2+(m-1)\beta}$ from block~$B_{m-1}$. The upper bound holds because the sum of any set of non-consecutive Fibonacci numbers with indices up to~$e$ is strictly less than~$F_{e+1}$ (a standard property of Zeckendorf representations).

Taking logarithms, both bounds are $\Theta(m\beta \cdot \log_{10}\varphi) = \Theta(m)$. For the 10-symbol alphabet ($\beta=11$), this gives approximately $2.3$ digits per symbol.
\end{proof}

\subsection{Optimality}

\begin{theorem}[Asymptotic optimality]\label{thm:optimal}
For $k \ge 2$, the block-partition encoding has asymptotically optimal code length: $L(m) = \Theta(m)$, matching the information-theoretic lower bound.
\end{theorem}

\begin{proof}
There are $k^m$ distinct strings of length~$m$. Any injective map into~$\mathbb{N}$ must use numbers as large as $k^m - 1$, requiring at least $m \cdot \log_{10}(k)$ digits. The block-partition encoding achieves $L(m) = \Theta(m \cdot (k+1) \cdot \log_{10}\varphi)$. The ratio of the constant factors is
\[
\frac{(k+1)\,\log_{10}\varphi}{\log_{10} k}.
\]
For $k=10$, this is $11 \cdot 0.209 / 1 \approx 2.3$. Both bounds are linear in~$m$; the block-partition encoding incurs a constant-factor overhead of approximately~$2.3$ relative to an information-theoretically optimal encoding.
\end{proof}

\subsection{Exponential blowup of right-nested carryless pairing}\label{sec:blowup}

Rosko~\cite{Rosko2025} defines a carryless pairing function $\pi : \mathbb{N}^2 \to \mathbb{N}$. His paper presents it as a binary pairing and sketches a $k$-ary generalization, but does not develop a completed sequence encoding. The natural way to encode a length-$m$ sequence $(a_1, a_2, \ldots, a_m)$ using a binary pairing is right-nesting. Formally, define $t_1 = a_m$ and $t_{j+1} = \pi(a_{m-j},\, t_j)$ for $j = 1, \ldots, m-1$. The final code is~$t_{m-1}$, requiring $m-1$ pairings. Let $M_j = \max Z(t_j)$ denote the largest Fibonacci index in the support after~$j$ pairings.

To avoid the issue that $Z(0) = \varnothing$ (making Rosko's pairing degenerate for zero inputs---in particular $\pi(0,0)=0$), we adopt the standard shift-by-one convention: symbol codes are drawn from $\{1, 2, \ldots, k\}$ rather than $\{0,\ldots,k-1\}$.

\begin{proposition}[Exponential growth]\label{prop:blowup}
Let codes be drawn from $\{1,\ldots,k\}$. Right-nested encoding of a length-$m$ sequence via Rosko's carryless pairing produces codes whose worst-case digit count is $\Theta(2^m)$.
\end{proposition}

\begin{proof}
In Rosko's pairing $\pi(x,y)$, the output's Zeckendorf support occupies indices up to $B(x) + 2\,R(y) - 1$, where $B(x) = 2\,r(x)$ is the delimiter, $r(x) = \min\{e : F_e > x\}$ is the rank, and $R(y) = \max Z(y)$. Since all symbol codes are positive, $Z(a_i)$ is nonempty and $R(a_i)$ is well-defined for every symbol.

Over the fixed symbol set $\{1,\ldots,k\}$, define
\[
b_{\min} = 2 \cdot \min\bigl\{\, r(a) : a \in \{1,\ldots,k\} \,\bigr\}, \qquad
b_{\max} = 2 \cdot \max\bigl\{\, r(a) : a \in \{1,\ldots,k\} \,\bigr\}.
\]
At each nesting step, the odd band determines the growth:
\[
2\,M_j + b_{\min} - 1 \;\le\; M_{j+1} \;\le\; 2\,M_j + b_{\max} - 1.
\]
Both bounds are linear recurrences of the form $M_{j+1} = 2M_j + C$, solving to $A \cdot 2^j - C$ for a constant~$A$ determined by initial conditions. Therefore $M_j = \Theta(2^j)$. Since the digit count is $\Theta(M_j \cdot \log_{10}\varphi)$, the encoded number has $\Theta(2^j)$ digits. A length-$m$ sequence requires $m-1$ pairings, giving worst-case $\Theta(2^m)$ digit growth.
\end{proof}

For concreteness: with symbol code~5 (where $r(5)=6$, $B(5)=12$, $R(5)=5$), the recurrence gives $M_1=21$, $M_2=53$, $M_3=117$, $M_4=245$. A length-5 sequence requires 4~pairings, so the largest Fibonacci index is $M_4=245$. By comparison, the block-partition encoding of a length-5 string uses a maximum Fibonacci index of $5 \cdot 11 = 55$.

\section{Comparison of encodings}

\begin{table}[ht]
\centering
\begin{tabular}{ll}
\toprule
\textbf{Encoding} & \textbf{Code length growth} \\
\midrule
Prime-power (G\"odel, 1931) & Super-linear \\
Base-$k$ positional & Linear (optimal constant) \\
Right-nested carryless pairing (Rosko, 2025) & Exponential: $\Theta(2^m)$ (Prop.~\ref{prop:blowup}) \\
\textbf{Block partition (this paper)} & \textbf{Linear: $\Theta(m)$ (Thm.~\ref{thm:linear})} \\
\bottomrule
\end{tabular}

\vspace{1em}
\caption{Worst-case code length for a string of length~$m$. To our knowledge, the block-partition construction is the first Zeckendorf-based encoding to achieve linear code growth.}
\label{tab:comparison}
\end{table}

\section{Conclusion}

The block-partition Fibonacci encoding provides a simple, compact mapping from finite strings to natural numbers, with worst-case code length growing linearly in the string length. The construction rests on an elementary observation: partitioning the Fibonacci sequence into blocks with single-index gaps ensures that any selection of one Fibonacci number per block yields a valid Zeckendorf representation, and Zeckendorf's theorem then guarantees unique decodability.

For sequence encoding over the same Zeckendorf foundation, the block-partition construction achieves $\Theta(m)$ code length where the natural right-nested use of Rosko's binary carryless pairing produces $\Theta(2^m)$ code length (Proposition~\ref{prop:blowup}), an exponential improvement.

Whether the additive character of this encoding can be exploited for theory-internal purposes---such as establishing $\Delta_0$-definability of syntactic predicates in bounded arithmetic or enabling incompleteness proofs in systems lacking multiplication---is a question we leave to future work.


\end{document}